\def\year{2019}\relax
\documentclass[letterpaper]{article}
\usepackage{aaai19}
\usepackage{times}
\usepackage{helvet}
\usepackage{courier}
\usepackage{url}
\usepackage{graphicx}

\usepackage[table,x11names]{xcolor}
\usepackage{amsmath,amsfonts,amssymb,amsbsy,amsmath}
\usepackage{subfig}
\usepackage{caption}
\usepackage{algorithm,algorithmic}
\usepackage{tikz}

\newtheorem{defn}{Definition}
\newtheorem{thm}{Theorem}
\newenvironment{proof}{\paragraph{Proof:}}{\hfill$\square$}

\begin{document}
\title{Multigrid Backprojection Super--Resolution\\ and Deep Filter Visualization}
\author{Pablo Navarrete Michelini, Hanwen Liu and Dan Zhu\\
BOE Technology Group Inc., Ltd.\\
Beijing, China
}
\nocopyright
\maketitle
\begin{abstract}
We introduce a novel deep--learning architecture for image upscaling by large factors (e.g. $4\times$, $8\times$) based on examples of pristine high--resolution images. Our target is to reconstruct high--resolution images from their downscale versions. The proposed system performs a multi--level progressive upscaling, starting from small factors ($2\times$) and updating for higher factors ($4\times$ and $8\times$). The system is recursive as it repeats the same procedure at each level. It is also residual since we use the network to update the outputs of a classic upscaler. The network residuals are improved by Iterative Back--Projections (IBP) computed in the features of a convolutional network. To work in multiple levels we extend the standard back--projection algorithm using a recursion analogous to Multi--Grid algorithms commonly used as solvers of large systems of linear equations. We finally show how the network can be interpreted as a standard upsampling--and--filter upscaler with a space--variant filter that adapts to the geometry. This approach allows us to visualize how the network learns to upscale. Finally, our system reaches state of the art quality for models with relatively few number of parameters.
\end{abstract}

\section{Introduction}
\label{intro}
In this work, we focus on the problem of image upscaling using convolutional networks. Upscaling signals by integer factors (e.g. $2\times$, $3\times$) is understood in classical interpolation theory as two sequential processes: upsample (insert zeros) and filter \cite{JGProakis_2007a,SMallat_1998a}. Standard upscaler algorithms, such as Bicubic or Lanczos, find high--resolution images with a narrow frequency content by using fixed low--pass filters. Similar to the classic upscaling model, the image acquisition can be modeled as low-pass filtering a high resolution image and then downsample the result (drop pixels). In test scenarios often used in benchmarks we actually know the exact downscaling model, e.g. Bicubic downscaler. The Iterative Back--Projection (IBP) algorithm \cite{Irani_1991a} is often used to enforce the downscaling model for a given upscaler and get closer to the original image.

More advanced upscalers follow geometric principles to improve image quality. For example, \emph{edge--directed interpolation} uses adaptive filters to improve edge smoothness \cite{VRAlgazi_1991a,XLi_2001a}, or \emph{bandlet} methods use both adaptive upsampling and filtering \cite{SMallat_2007a}. More recently, machine learning has been able to use examples pairs of high and low resolution images to estimate the parameters of upscaling systems \cite{SCPark_2003a}. In some cases, the optimization approach of machine learning hides the connection with classical interpolation theory, e.g. sparse representation with dictionaries \cite{JYang_2008a,JYang_2010a}. In other cases, the adaptive filter approach is explicit, e.g. RAISR \cite{DBLP:journals/corr/RomanoIM16}.

Upscaling using convolutional networks started with SRCNN \cite{CDong_2014a,CDong_2015a} motivated by the success of deep--learning methods in image classification tasks \cite{YLeCun_2015a} and establishing a strong connection with sparse coding methods \cite{JYang_2008a,JYang_2010a}. SRCNN has later been improved, most notably by EDSR \cite{Lim_2017_CVPR_Workshops} and DBPN \cite{DBPN2018}. Our system shares the convolutional network approach but follows a different motivation. Namely, we aim to reveal a strong connection between convolutional--networks and classical image upscaling. By doing so, we can recover the classic interpretation of upsampling and filter and visualize what is the network doing pixel by pixel. Thus, we aim to prove that convolutional networks are a natural and convenient choice for Super--Resolution (SR) tasks.

Our main contributions are:
\begin{itemize}
    \item We extend the IBP algorithm to a \textbf{multi--level IBP}.
    \item We prove that our algorithm \textbf{works as well as classic IBP}, based on an unrealistic model.
    \item We introduce a \textbf{new network architecture} that overcomes the unrealistic model and allows us to learn both upscaling and downscaling.
    \item We introduce a novel \textbf{algorithm to analyze} the linear components of the network.
    \item We show how to \textbf{interpret the network} as a standard upscaler with adaptive filters.
\end{itemize}

\section{Related Work}
\label{sec:related}
We consider the following two architectures as the most similar to our system:\pagebreak
\begin{itemize}
    \item \textbf{Multi--Scale Laplacian Super--Resolution (MSLapSR)} \cite{MSLapSRN}: Our system is inspired by MSLapSR to progressively upscale images using a classic upscaler and network updates, before back--projections that improve network updates in lower--resolutions. Both MSLapSR and our systems include \emph{analysis} and \emph{synthesis} networks to convert images to latent space and vice versa. Both systems share parameters at each scale. Our system differs mostly in the use of back--projections, that cannot be removed to recover MSLapSR because of the particular structure of our \emph{upscaler} network module.
    \item \textbf{Deep Back--Projection Network} (DBPN) \cite{DBPN2018}: To the extent of our knowledge, this is the first reference to use IBP in a network architecture. It is also the state--of--art in terms of image quality, surpassing EDSR\cite{Lim_2017_CVPR_Workshops}, former winner of NTIRE 2017 SR Challenge \cite{Timofte_2017_CVPR_Workshops}. Their approach to use back--projections is different than ours because: first, their system is not multi--scale; and second, they iterate down and up projections. Our multi--scale architecture requires less number of parameters, because we reuse modules at every scale, and it is more flexible, because many upscaling factors can be achieved with the same modules. Also, our system is built upon an algorithm that is proven to converge, whereas, to the extent of our knowledge, mixing up and down projections has no convergence guarantees.
\end{itemize}
\begin{algorithm*}[th]
    \centering
    \begin{tabular}{ll}
        $\boldsymbol{MGBP}(X,\mu,L)$: & $\boldsymbol{BP^{\mu}_{k}}(u, Y_1,\ldots,Y_{k-1})$: \\

        \resizebox{.5\textwidth}{!}{
            \begin{minipage}{.6\textwidth}
            \begin{algorithmic}[1]
                \REQUIRE Input image $X$.
                \REQUIRE Integers $\mu\geqslant 0$ and $L\geqslant 1$.
                \ENSURE Images $Y_k$, $k = 2,\ldots,L$.

                \STATE $Y_1 = X$
                \FOR{$k = 2,\ldots,L$}
                    \STATE $u = (Y_{k-1}\uparrow s) * p$
                    \STATE $Y_{k} = BP^{\mu}_{k}\left(u, Y_1,\ldots,Y_{k-1} \right)$
                \ENDFOR
            \end{algorithmic}
        \end{minipage}
        }
        &
        \resizebox{.5\textwidth}{!}{
            \begin{minipage}{0.6\textwidth}
            \begin{algorithmic}[1]
                \REQUIRE Image $u$, level index $k$, number of steps $\mu$.
                \REQUIRE Images $Y_1,\ldots,Y_{k-1}$ (only for $k>1$).
                \ENSURE Updated image $u$

                \IF{$k > 1$}
                    \FOR{$step = 1,\ldots,\mu$}
                        \STATE $d = BP^{\mu}_{k-1}\left( {\color{blue}(u * g)\downarrow s}, Y_1,\ldots,Y_{k-2} \right)$
                        \STATE $u = u + {\color{red}(Y_{k-1} - d)\uparrow s * p}$
                    \ENDFOR
                \ENDIF
            \end{algorithmic}
            \end{minipage}
        }
    \end{tabular}
    \caption{Multi--Grid Back--Projection (MGBP)} \label{tab:alg}
    \label{alg:mgbp_classic}
\end{algorithm*}
\begin{figure*}[t]
  \centering
  \includegraphics[width=\linewidth]{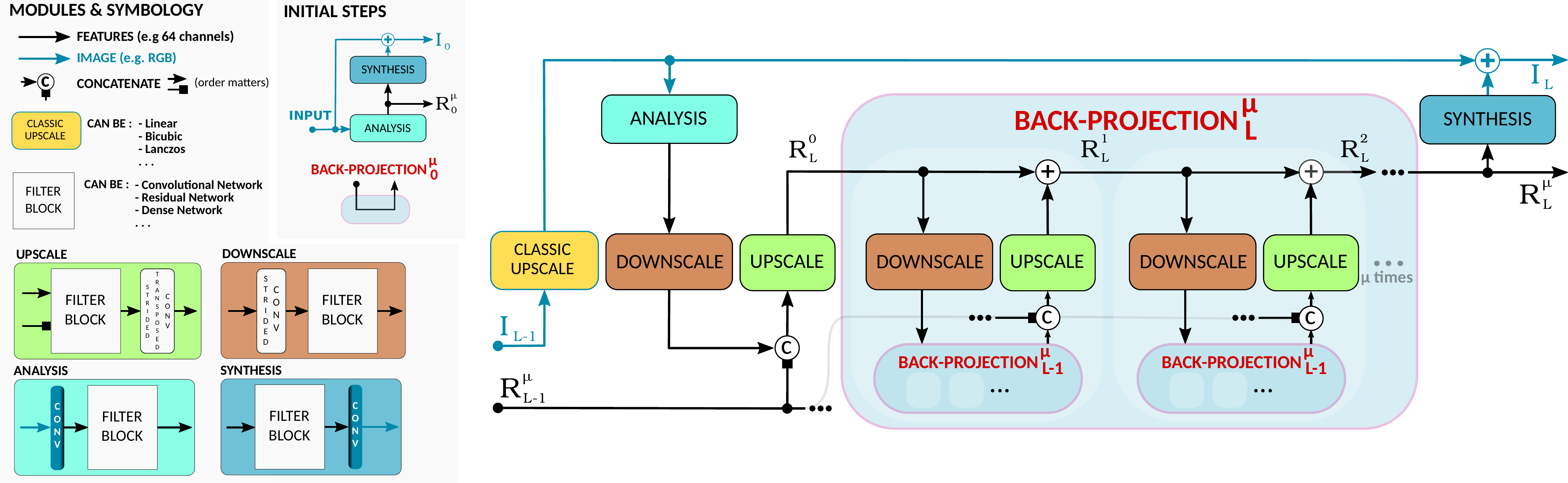}
  \caption{
  Multigrid Back--Projection (MGBP) system architecture. Lines indicate a number of image channels moving to different processing blocks, and line colors indicate the number of channels. Blue represents $3$ channels for RGB images, and black represents the number of features managed by convolutional networks (latent space). At level $k$ the current upscaled image is indicated as $I_k$ and residuals in latent space are indicated as $R_k$. We use \emph{Analysis} and \emph{Synthesis} modules to transfer images into latent space and vice versa. In \emph{Initial Steps} we obtain the first pair of output image $I_0$ and residual $R_0$, from the input image. Then, the main diagram shows how to obtain the pair $I_L$ and $R_L$ from the previous pair $I_{L−1}$ and $R_{L-1}$. In red we show the Back--Projection module, which repeats mu back--projection steps recursively.
  \label{fig:sys_diagram}}
\end{figure*}
\begin{figure*}[h]
  \centering
  \includegraphics[width=\linewidth]{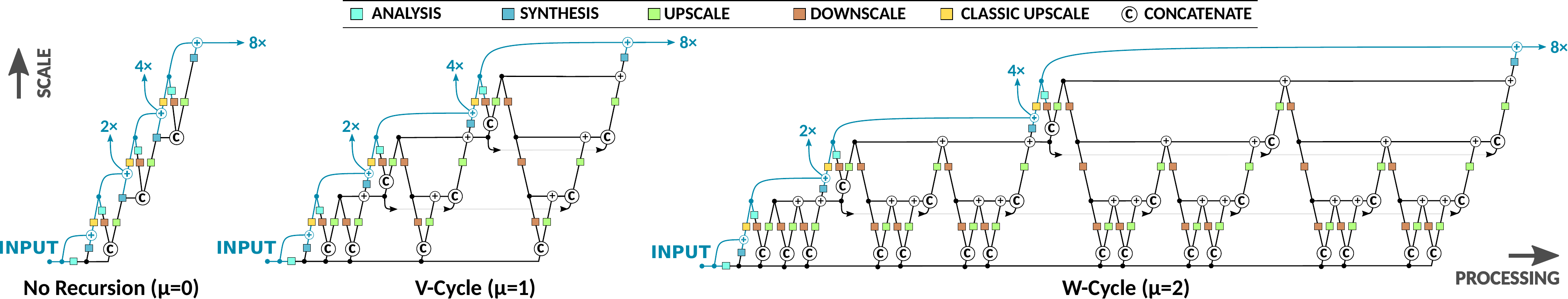}
  \caption{Multi--Grid Back--Projection (MGBP) recursion unfold from Figure \ref{fig:sys_diagram} for different values of $\mu$, and three levels to output $2\times$, $4\times$ and $8\times$ upscale images. Our system uses a recursion analogous to the Full--Multigrid algorithm to solve linear equations and leads to the well known workflows V--cycle for $\mu=1$ and W--cycles for $\mu=2$ \cite{UTrottenberg_2000a}. After every upscaling step, the MGBP recursion sends the output down to every lower--resolution level in order to validate the downscaling model. The corrections are back-projected to higher resolutions. \label{fig:sys_workflow}}
\end{figure*}
Our contributions add to the following lines of research:
\begin{itemize}
    \item \textbf{Recursive Architectures}: Extensive work has been done on recursive CNN architectures inspired on iterative procedures. Our system belongs to this line of work since we propose an architecture based on IBP and multi--grid iterations. Of particular interest is the research along this line related to SR and enhancements problems. In \cite{Kruse_2017_ICCV}, for example, a \emph{half quadratic splitting} iteration with a CNN denoiser prior is proposed for image denoising, deblurring and SR. Most recently, \cite{DBLP:journals/corr/abs-1803-05215} take a similar approach for image demosaicking. Next, \cite{DBLP:journals/corr/abs-1804-03368} propose to learn an optimizer for image deconvolution by recurrently incorporating CNNs into a gradient descent scheme, which could be complementary to our approach. In \cite{Zhang_2017_CVPR}, an energy minimization iteration is performed using a CNN model for image deconvolution, which is close to the SR problem. Finally, the work of \cite{DBLP:journals/corr/DiamondSHW17} follows a similar motivation than ours by proposing a framework to incorporate knowledge of the image formation into CNNs, applied to denoising, deblurring and compressed sensing.

    All these approaches differ to ours in the fact that our recursion iterates back and forth between different resolutions. In numerical methods, such iterations have been used by two types of linear equation solvers: \emph{multi--grid} and \emph{domain decomposition} methods \cite{UTrottenberg_2000a,Widlund_DD}. We follow the multi--grid approach, that can use inductive arguments to study convergence, and leads to specific processing workflows to move intermediate results between scales (see for example the W--cycle in Figure \ref{fig:sys_workflow}). Systems like MS--DenseNet \cite{DBLP:journals/corr/HuangCLWMW17} also move back and forth between scales with more simplified workflows but they were not designed based on classical methods. The connection to classical methods gives us a justification of these workflows that, otherwise, would be arbitrary (e.g. why traversing scales with V or W workflows in Figure \ref{fig:sys_workflow})? which workflow is better?) We will show, for example, how different numbers of back--projections (related to depth) make outputs sharper, same as in IBP. Our main contribution here is to devise a new algorithm that introduces the multi--grid recursion into IBP. It is a different algorithm than IBP, that we prove to converge at the same rate, and then extend to a network architecture. Finally, this effort pays back since our recursion works well in experiments, with no other method reaching the same quality with the same number of parameters.
    \item \textbf{Network Visualization}:
    A major direction of research in deep--learning is how to visualize the inner processing of a given architecture \cite{Zhang_DLVisualInterpretability}. Among these, a line of research on \emph{feature visualization} studies what does a network detect \cite{olah2017_featurevis}. Feature visualization can give example inputs that cause desired behaviors, separating image areas causing behavior from those that only relate to the causes. Another line of work on \emph{attribution} studies how does a network assembles these individual pieces to arrive at later decisions, or why these decisions were made \cite{olah2018_buildingblocks}. Our visualization technique belongs to the latter because we can show how the input pixels are assembled into a particular output pixel. We target the SR problem where there is extensive knowledge of non--adaptive filters (e.g. linear, bicubic, etc.) built upon signal processing theory. The main novelty of our technique is that it provides an alternative system to replace the network for a given input. The new system generates the exact same outputs of the network from the same input images but, unlike the network, it is fully interpretable in the sense that we know what to expect from its parameters.
\end{itemize}

\section{Multigrid Backprojections}
\label{sec:MGBP}
A simple and common model for the downscaling process is
\begin{equation}
    X = (Y*g)\downarrow s \;, \label{eq:model_down}
\end{equation}
where $Y$ is the high--resolution source, $X$ is the low--resolution result, $g$ is a blurring kernel and $\downarrow s$ is a downsampling by factor $s$.

Model \eqref{eq:model_down} gives additional information about the unknown high--resolution image and narrows down the search space from all possible images to images that downscaled with model \eqref{eq:model_down} recover the low--resolution input image. This is the motivation behind the classic IBP algorithm \cite{Irani_1991a}. Given model \eqref{eq:model_down} and an upscaled image $Y$, the IBP algorithm iterates:
\begin{eqnarray}
    e(Y_k) & = & X-(Y_k*g)\downarrow s \\
    Y_{k+1} & = & Y_k+e(Y_k)\uparrow s * p \;.
\end{eqnarray}
Here, $e(Y_k)$ is the mismatch error at low--resolution, $g$ and $p$ are blurring and upscaling filters, respectively. The iteration is proven to converge, to enforce model \eqref{eq:model_down} at exponential rate \cite{Irani_1991a}.

To make IBP work for multiple scales we change model \eqref{eq:model_down} to:
\begin{equation}
    X = (\cdots ((Y \underbrace{* g)\downarrow s * g)\downarrow s \cdots *g)\downarrow s}_{L \; \text{times}} \;. \label{eq:mgbp_model}
\end{equation}
This is not a common downscaling procedure in practice and might be the reason why multi--scale IBP has not been considered yet. We will later replace the downscaling by a network so that model \eqref{eq:mgbp_model} becomes flexible and is able to learn a direct downscaling or even more complex models.

Upscaling images with IBP is a two--step process: first, upscale an image; and second, improve it with IBP. This is reminiscent of the way a Full--Multigrid algorithm solves linear equations \cite{UTrottenberg_2000a}. This is: first, find an approximate solution; and second, improve it by solving an equation for the approximation error. Both IBP and Multigrid iterate between different scales, but IBP only uses two levels whereas Multigrid recursively move to coarser grids. We use the same strategy as in Multigrid to define a so--called \emph{Multi--Grid Back--Projection} algorithm as shown in Algorithm \ref{alg:mgbp_classic}.
Here, back--projections recursively return to the lowest--resolution enforcing the downscaling model at each scale.

For $L=2$ the MGBP algorithm and model \eqref{eq:mgbp_model} are equivalent to the original IBP \cite{Irani_1991a} and model \eqref{eq:model_down}, and thus converges at exponential rate. In section A of the supplementary material we prove convergence for $L>2$. Basically, the algorithm inherits the exponential rate convergence from the two--level case through the recursion in Algorithm \ref{alg:mgbp_classic}.

\section{Network Architecture}
\label{sec:network}
We convert the Multi--Grid Back--Projection algorithm into a network structure as follows:
\begin{itemize}
    \item Step 1: Use a classic method to upscale a low resolution image.
    \item Step 2: Transfer the upscale image $I$ into latent space using a network $\text{Analysis}(I)$.
    \item Step 3: In latent space we apply the recurrence in Algorithm \ref{alg:mgbp_classic} by changing:
    \begin{itemize}
        \item ${\color{blue}(u * g)\downarrow s}$ into a network ${\color{blue}\text{Downscale}(u)}$.
        \item ${\color{red}(Y_{k-1} - d)\uparrow s * p}$ into a network ${\color{red}\text{Upscale}(\left[Y_{k-1}, d\right])}$. Where $\left[Y_{k-1}, d\right]$ is the concatenation of features and replaces the subtraction. Thus, the \emph{Upscaler} network receives double the number of features in the input compared to the output.
    \end{itemize}
\end{itemize}
Figure \ref{fig:sys_diagram} shows the definition of our network architecture. The unfolded recursion is shown in Figure \ref{fig:sys_workflow} for three different values of $\mu$.

\section{Deep Filter Visualization}
\label{sec:visualization}
Deep Learning architectures are highly non--linear, although much of their internal structure is linear (e.g. convolutions). We want to study the overall effect of the linear structure of the network. The general procedure is shown in Figure \ref{fig:vis_diagram} and is as follows:
\begin{itemize}
    \item An input image $X$ passes through all layers of the network, and outputs an image $Y$.
    \item At each non-linear layer we record how much did the input change (layer gain).
    \item We replace (freeze) all non--linearities by the fix gain previously recorded. The overall system becomes linear, such that $Y = FX+R$.
    \item We obtain the \emph{effective residual} $R$ with an input $X=0$ in the activation frozen system.
    \item We obtain the \emph{effective filter} for a particular input pixel by using a $\delta$ input centered at the pixel location, and subtract the residual $R$ from the output. In linear system terminology we are computing the \emph{impulse response} of the system \cite{JGProakis_2007a}.
\end{itemize}

We can obtain explicit formulas for $F$ and $R$ if we consider a model of convolutional network as a sequence of linear and nonlinear layers:
\begin{equation}
    z_n = W_n x_{n-1} + b_n \quad\quad\text{and}\quad\quad x_n = \sigma\left(z_n\right) \;, \label{eq:cnn_model}
\end{equation}
where $x_n$ and $z_n$ are vectorized features at layer $n$, after and before activations, respectively. The parameters of the network are the biases $b_n$ and the sparse matrices $W_n$ representing the convolutional operators (also applies for strided and transposed convolutions). For a given input image $X$, the input of the model is $x_0=\text{vec}(X)$ (vectorized image). The output image $Y$, after $n$ convolutional layers, is $\text{vec}^{-1}(x_n)$.
\begin{table*}[ht]
  \caption{\textbf{Quantitative evaluation} of different SR methods. In {\color{red} red} color we show the best result, in {\color{blue} blue} the top--$2$ and {\color{brown} brown} the top--$3$ results for each column and scale. Methods are ordered by increasing number of parameters.}
  \label{tab:quantitative}
  \centering
\resizebox{\linewidth}{!}{
  \begin{tabular}{lcccccccccc}\hline
    Algorithm                           & $s$  & $par$ & \multicolumn{2}{c}{Set14} & \multicolumn{2}{c}{BSDS100} & \multicolumn{2}{c}{Urban100} & \multicolumn{2}{c}{Manga109} \\
                                        &      & $[M]$ & PSNR & SSIM & PSNR & SSIM & PSNR & SSIM & PSNR & SSIM \\ \hline
    Bicubic                             & 2 & -- & $30.34$ & $0.870$ & $29.56$ & $0.844$ & $26.88$ & $0.841$ & $30.84$ & $0.935$ \\
    A+~\cite{Timofte_2014a}             & 2 & -- & $32.40$ & $0.906$ & $31.22$ & $0.887$ & $29.23$ & $0.894$ & $35.33$ & $0.967$ \\
    FSRCNN~\cite{Dong_2016a}            & 2 & $0.01$ & $32.73$ & $0.909$ & $31.51$ & $0.891$ & $29.87$ & $0.901$ & $36.62$ & $0.971$ \\
    SRCNN~\cite{Dong_2014a}             & 2 & $0.06$ & $32.29$ & $0.903$ & $31.36$ & $0.888$ & $29.52$ & $0.895$ & $35.72$ & $0.968$ \\
    MSLapSRN~\cite{MSLapSRN}            & 2 & $0.22$ & {\color{brown}$33.28$} & {\color{blue}$0.915$} & {\color{brown}$32.05$} & {\color{brown}$0.898$} & $31.15$ & $0.919$ & $37.78$ & {\color{brown}$0.976$} \\
    \rowcolor{lightgray} Our            & 2 & $0.28$ & $33.27$ & {\color{blue}$0.915$} & $31.99$ & $0.897$ & {\color{brown}$31.37$} & {\color{brown}$0.920$} & {\color{brown}$37.92$} & {\color{brown}$0.976$} \\
    VDSR~\cite{Kim_2016_VDSR}           & 2 & $0.67$ & $32.97$ & $0.913$ & $31.90$ & $0.896$ & $30.77$ & $0.914$ & $37.16$ & $0.974$ \\
    LapSRN~\cite{LapSRN}                & 2 & $0.81$ & $33.08$ & $0.913$ & $31.80$ & $0.895$ & $30.41$ & $0.910$ & $37.27$ & $0.974$ \\
    DRCN~\cite{Kim_2016_DRCN}           & 2 & $1.78$ & $32.98$ & $0.913$ & $31.85$ & $0.894$ & $30.76$ & $0.913$ & $37.57$ & $0.973$ \\
    D-DBPN~\cite{DBPN2018}              & 2 & $5.95$ & {\color{blue}$33.85$} & {\color{red}$0.919$} & {\color{blue}$32.27$} & {\color{blue}$0.900$} & {\color{blue}$32.70$} & {\color{blue}$0.931$} & {\color{blue}$39.10$} & {\color{red}$0.978$} \\
    EDSR~\cite{Lim_2017_CVPR_Workshops} & 2 & $40.7$ & {\color{red}$33.92$} & {\color{red}$0.919$} & {\color{red}$32.32$} & {\color{red}$0.901$} & {\color{red}$32.93$} & {\color{red}$0.935$} & {\color{red}$39.10$} & {\color{blue}$0.977$} \\
    \noalign{\smallskip}\hline\noalign{\smallskip}
    Bicubic                             & 4 & -- & $26.10$ & $0.704$ & $25.96$ & $0.669$ & $23.15$ & $0.659$ & $24.92$ & $0.789$ \\
    A+~\cite{Timofte_2014a}             & 4 & -- & $27.43$ & $0.752$ & $26.82$ & $0.710$ & $24.34$ & $0.720$ & $27.02$ & $0.850$ \\
    FSRCNN~\cite{Dong_2016a}            & 4 & $0.01$ & $27.70$ & $0.756$ & $26.97$ & $0.714$ & $24.61$ & $0.727$ & $27.89$ & $0.859$ \\
    SRCNN~\cite{Dong_2014a}             & 4 & $0.06$ & $27.61$ & $0.754$ & $26.91$ & $0.712$ & $24.53$ & $0.724$ & $27.66$ & $0.858$ \\
    MSLapSRN~\cite{MSLapSRN}            & 4 & $0.22$ & $28.26$ & $0.774$ & {\color{brown}$27.43$} & $0.731$ & $25.51$ & $0.768$ & $29.54$ & $0.897$ \\
    \rowcolor{lightgray} Our            & 4 & $0.28$ & {\color{brown}$28.43$} & {\color{brown}$0.778$} & $27.42$ & {\color{brown}$0.732$} & {\color{brown}$25.70$} & {\color{brown}$0.774$} & {\color{brown}$30.07$} & {\color{brown}$0.904$} \\
    VDSR~\cite{Kim_2016_VDSR}           & 4 & $0.67$ & $28.03$ & $0.770$ & $27.29$ & $0.726$ & $25.18$ & $0.753$ & $28.82$ & $0.886$ \\
    LapSRN~\cite{LapSRN}                & 4 & $0.81$ & $28.19$ & $0.772$ & $27.32$ & $0.728$ & $25.21$ & $0.756$ & $29.09$ & $0.890$ \\
    DRCN~\cite{Kim_2016_DRCN}           & 4 & $1.78$ & $28.04$ & $0.770$ & $27.24$ & $0.724$ & $25.14$ & $0.752$ & $28.97$ & $0.886$ \\
    D-DBPN~\cite{DBPN2018}              & 4 & $10.4$ & {\color{red}$28.82$} & {\color{blue}$0.786$} & {\color{red}$27.72$} & {\color{blue}$0.740$} & {\color{blue}$26.54$} & {\color{blue}$0.795$} & {\color{red}$31.18$} & {\color{blue}$0.914$} \\
    EDSR~\cite{Lim_2017_CVPR_Workshops} & 4	& $43.1$ & {\color{blue}$28.80$} & {\color{red}$0.788$} & {\color{blue}$27.71$} & {\color{red}$0.742$} & {\color{red}$26.64$} & {\color{red}$0.803$} & {\color{blue}$31.02$} & {\color{red}$0.915$}\\
    \noalign{\smallskip}\hline\noalign{\smallskip}
    Bicubic                             & 8 & -- & $23.19$ & $0.568$ & $23.67$ & $0.547$ & $20.74$ & $0.516$ & $21.47$ & $0.647$ \\
    A+~\cite{Timofte_2014a}             & 8 & -- & $23.98$ & $0.597$ & $24.20$ & $0.568$ & $21.37$ & $0.545$ & $22.39$ & $0.680$ \\
    FSRCNN~\cite{Dong_2016a}            & 8 & $0.01$ & $23.93$ & $0.592$ & $24.21$ & $0.567$ & $21.32$ & $0.537$ & $22.39$ & $0.672$ \\
    SRCNN~\cite{Dong_2014a}             & 8 & $0.06$ & $23.85$ & $0.593$ & $24.13$ & $0.565$ & $21.29$ & $0.543$ & $22.37$ & $0.682$ \\
    MSLapSRN~\cite{MSLapSRN}            & 8 & $0.22$ & $24.57$ & $0.629$ & $24.65$ & {\color{brown}$0.592$} & $22.06$ & $0.598$ & $23.90$ & $0.759$ \\
    \rowcolor{lightgray} Our            & 8 & $0.28$ & {\color{brown}$24.82$} & {\color{brown}$0.635$} & {\color{brown}$24.67$} & {\color{brown}$0.592$} & {\color{brown}$22.21$} & {\color{brown}$0.603$} & {\color{brown}$24.12$} & {\color{brown}$0.765$} \\
    VDSR~\cite{Kim_2016_VDSR}           & 8 & $0.67$ & $24.21$ & $0.609$ & $24.37$ & $0.576$ & $21.54$ & $0.560$ & $22.83$ & $0.707$ \\
    LapSRN~\cite{LapSRN}                & 8 & $0.81$ & $24.44$ & $0.623$ & $24.54$ & $0.586$ & $21.81$ & $0.582$ & $23.39$ & $0.735$ \\
    D-DBPN~\cite{DBPN2018}              & 8 & $23.2$ & {\color{red}$25.13$} & {\color{red}$0.648$} & {\color{red}$24.88$} & {\color{red}$0.601$} & {\color{red}$22.83$} & {\color{red}$0.622$} & {\color{red}$25.30$} & {\color{red}$0.799$} \\
    EDSR~\cite{Lim_2017_CVPR_Workshops} & 8 & $43.1$ & {\color{blue}$24.94$} & {\color{blue}$0.640$} & {\color{blue}$24.80$} & {\color{blue}$0.596$} & {\color{blue}$22.47$} & {\color{blue}$0.620$} & {\color{blue}$ 24.58$} & {\color{blue}$0.778$} \\
  \end{tabular}
}
\end{table*}

\begin{defn}[Activation Gain]
    The \emph{gain} of an activation function $\sigma$ is given by $G(x)[i]=\sigma(x[i])/x[i]$ and equal to $1$ if $x[i]=0$.
\end{defn}
\begin{figure}
  \centering
  \includegraphics[width=\linewidth]{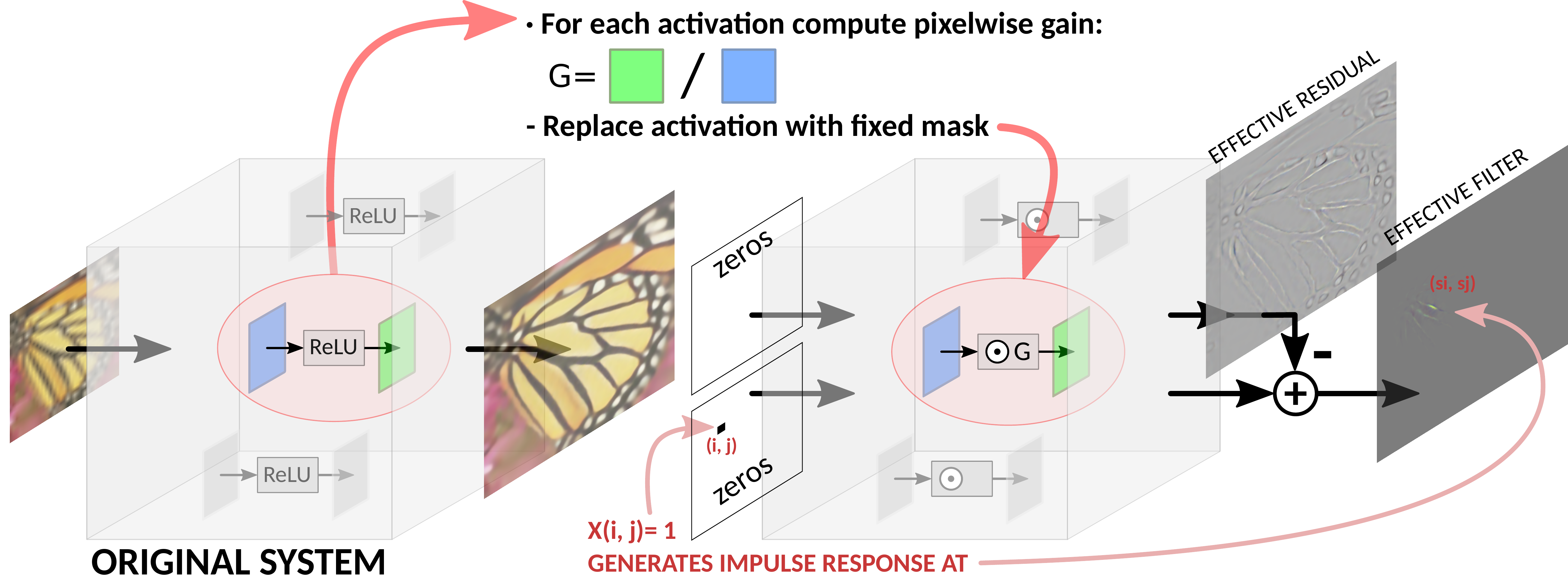}
  \caption{Activation freezing procedure to convert a network into a linear system. We can perform impulse response analysis to study the overall filter effect at each pixel location. \label{fig:vis_diagram}}
\end{figure}
\begin{figure}
  \centering
  \includegraphics[width=\linewidth]{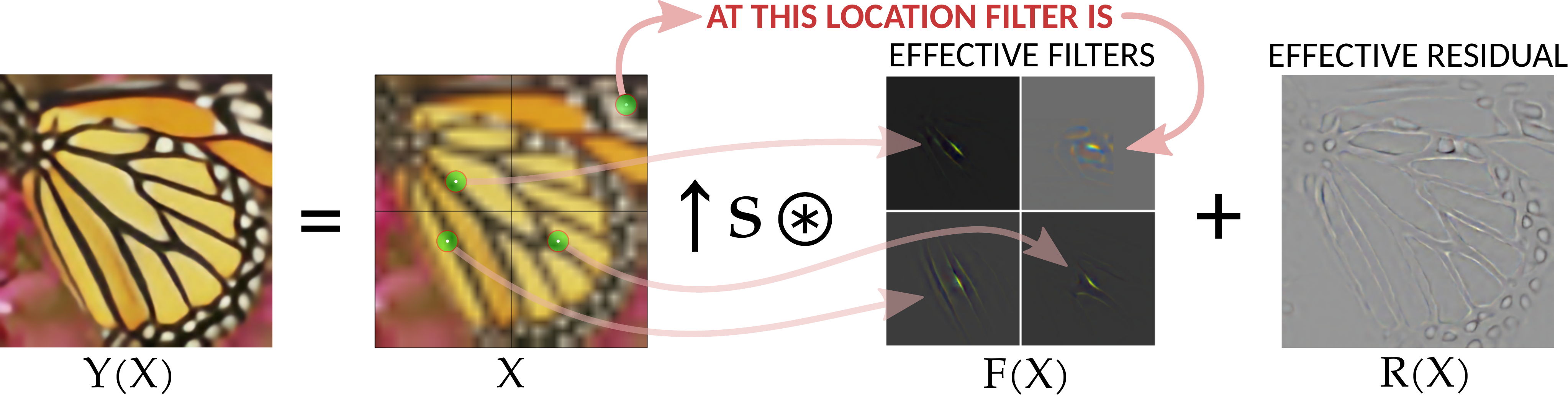}
  \caption{Interpretation of the upscaling network as a standard upscaler with adaptive filters. \label{fig:vis_interpretation}}
\end{figure}

\begin{thm}[Activation Freeze] \label{thm:act_freeze}
Let $\hat{W}_n = G(z_n) W_n$ and $\hat{b}_n = G(z_n) b_n$. Let
\begin{equation}
    Q_0 = I \;, \quad Q_i = \prod_{k=n-i+1}^n \hat{W}_k, \quad\text{for } i=1,\ldots,n \;.
\end{equation}
The output of the convolutional network is given by $x_n = F \; x_0 + R$ , where $F = Q_n$ is the \textbf{effective filter} and $R = Q \ast \hat{b} = \sum_{k=0}^n Q_k \hat{b}_{n-k}$ is the \textbf{effective residual}.
\end{thm}
The proof of the theorem is a direct consequence of $\sigma(x) = G(x) x$ and expansion of \eqref{eq:cnn_model}. Although this theorem only applies to sequential networks, it helps to show that the overall effective filter depends on all convolutional filters as well as biases (only through activations). Similarly, the effective residual depends on both convolutional filters and biases (both explicitly and through activations).

For the sake of simplicity, here we use our visualization technique from input to outputs of the network. Nevertheless, we note that we could start at any layer and stop at any posterior layer to study attributions within the network.

\section{Experiments}
\label{sec:experiments}
We use one \textbf{single configuration} to test our system for $2\times$, $4\times$, and $8\times$ upscaling. We configure the \emph{Analysis}, \emph{Synthesis}, \emph{Upscale} and \emph{Downscale} modules in Figure \ref{fig:sys_diagram} using $4$--layer dense networks \cite{huang2017densely}  as filter--blocks. We use $48$ features and growth rate $16$ within dense networks. For classic upscaler we start with Bicubic and the upscaling filters are set as parameters to learn during training.

For \textbf{training} the system we fix the number of back--projections to $\mu=2$ (W--cycle according to Figure \ref{fig:sys_workflow}). We train the system for $8\times$ upscaler with the multi--scale loss function introduced in \cite{MSLapSRN}:
\begin{equation}
    \mathcal{L}(Y,X;\theta)=\sum_{L=1,2,3} \sum_{k=1}^L \mathbb{E}\left[ ||Y^L_k-X_k||_1 \right] \;.
\end{equation}
We train our system with Adam optimizer and a learning rate initialized as $10^{-3}$ and square root decay. We use $128\times128$ patches (pieces of images) with batch size $16$. The patches were sampled randomly from datasets DIV2K and Flickr2K, containing photographs of general natural scenes.

\subsection{Comparisons}
\label{ssec:comparisons}
\begin{figure*}[h!]
  \centering
  \includegraphics[width=\linewidth]{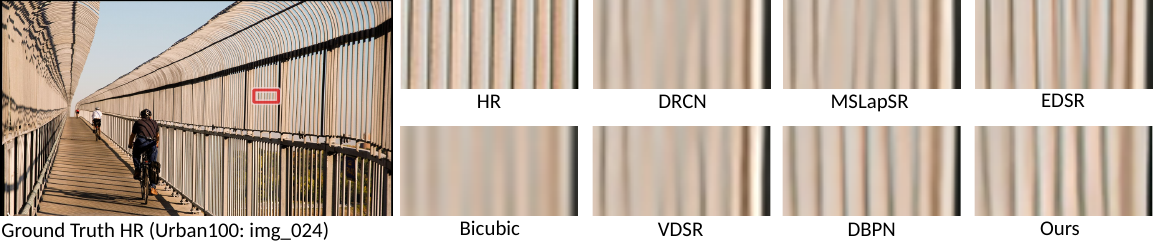}
  \includegraphics[width=\linewidth]{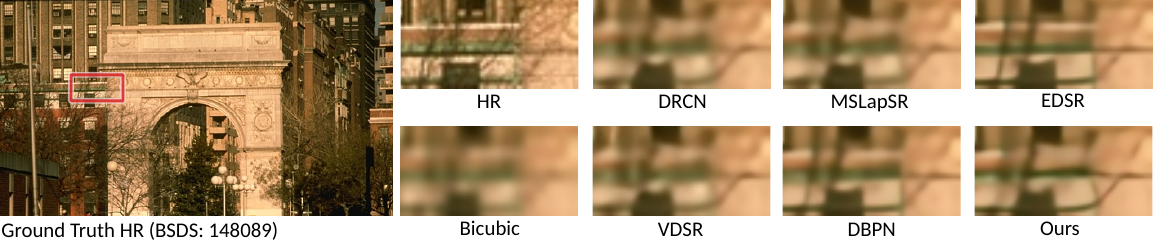}
  \includegraphics[width=\linewidth]{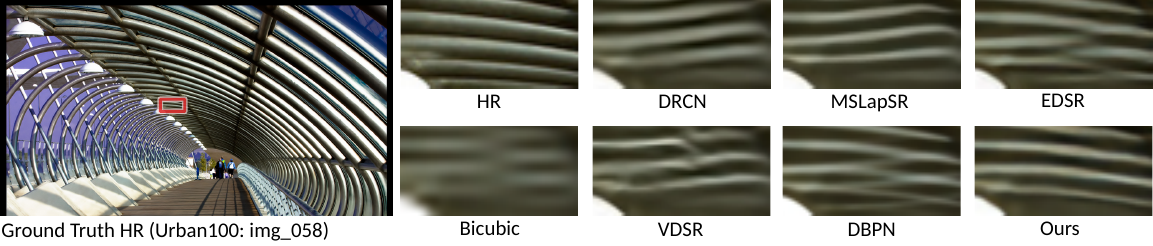}
  \caption{Perceptual evaluation of different SR methods for $4\times$ upscaling. \label{fig:perceptual_4x}}
\end{figure*}
\begin{figure*}[h!]
  \centering
  \includegraphics[width=0.9\linewidth]{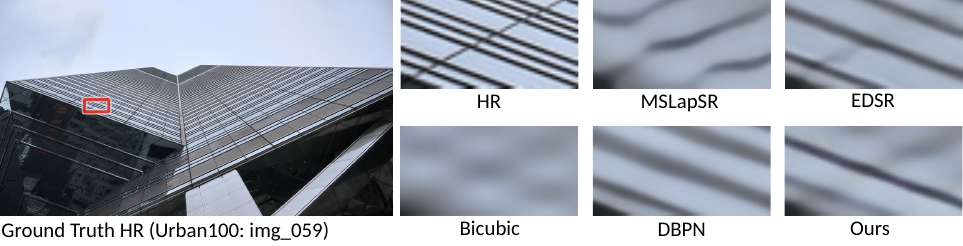}
  \caption{Perceptual evaluation of different SR methods for $8\times$ upscaling. \label{fig:perceptual_8x}}
\end{figure*}
In Table \ref{tab:quantitative} we compare PSNR and SSIM values for different methods. The two evaluation metrics measure the difference between an upscaler output and the original high-resolution image. Higher values are better in both cases. Roughly speaking, PSNR (range $0$ to $\infty$) is a log--scale version of mean--square--error and SSIM (range $0$ to $1$) uses image statistics to better correlate with human perception. Full expressions are as follows:
\begin{align}
    PSNR(X,Y) & = 10 \cdot \log_{10}\left(\frac{255^2}{MSE}\right) \;,\\
    SSIM(X,Y) & =\frac{(2\mu_X\mu_Y+c_1)(2\sigma_{XY}+c_2)}{(\mu_X^2+\mu_Y^2+c_1)(\sigma_X^2+\sigma_Y^2+c_2)} \;,
\end{align}
where $MSE=mean((X-Y)^2)$ is the mean square error of $X$ and $Y$; $\mu_X$ and $\mu_Y$ are the averages of $X$ and $Y$, respectively; $\sigma_X^2$ and $\sigma_Y^2$ are the variances of $X$ and $Y$, respectively; $\sigma _{XY}$ is the covariance of X and Y; $c_1=6.5025$ and $c_2=58.5225$.

Our method is outperformed only by EDSR and DBPN, that use $20$ to more than $100$ times the number of parameters of our system. Among systems with less than $2$ million parameters, our system obtains better results, only matched by MSLapSR for $2\times$ upscaling.

In Figure \ref{fig:scatter} we better visualize the difference in complexity between different systems. It is apparent from these results that the size of the system matters to improve PSNR values, with EDSR and DBPN far from other methods, but the cost can be overwhelming in performance. Our system clearly improves the state of the art for systems with less than two million parameters, with better quality and significantly less parameters.

\begin{figure}[h]
  \centering
  \includegraphics[width=\linewidth]{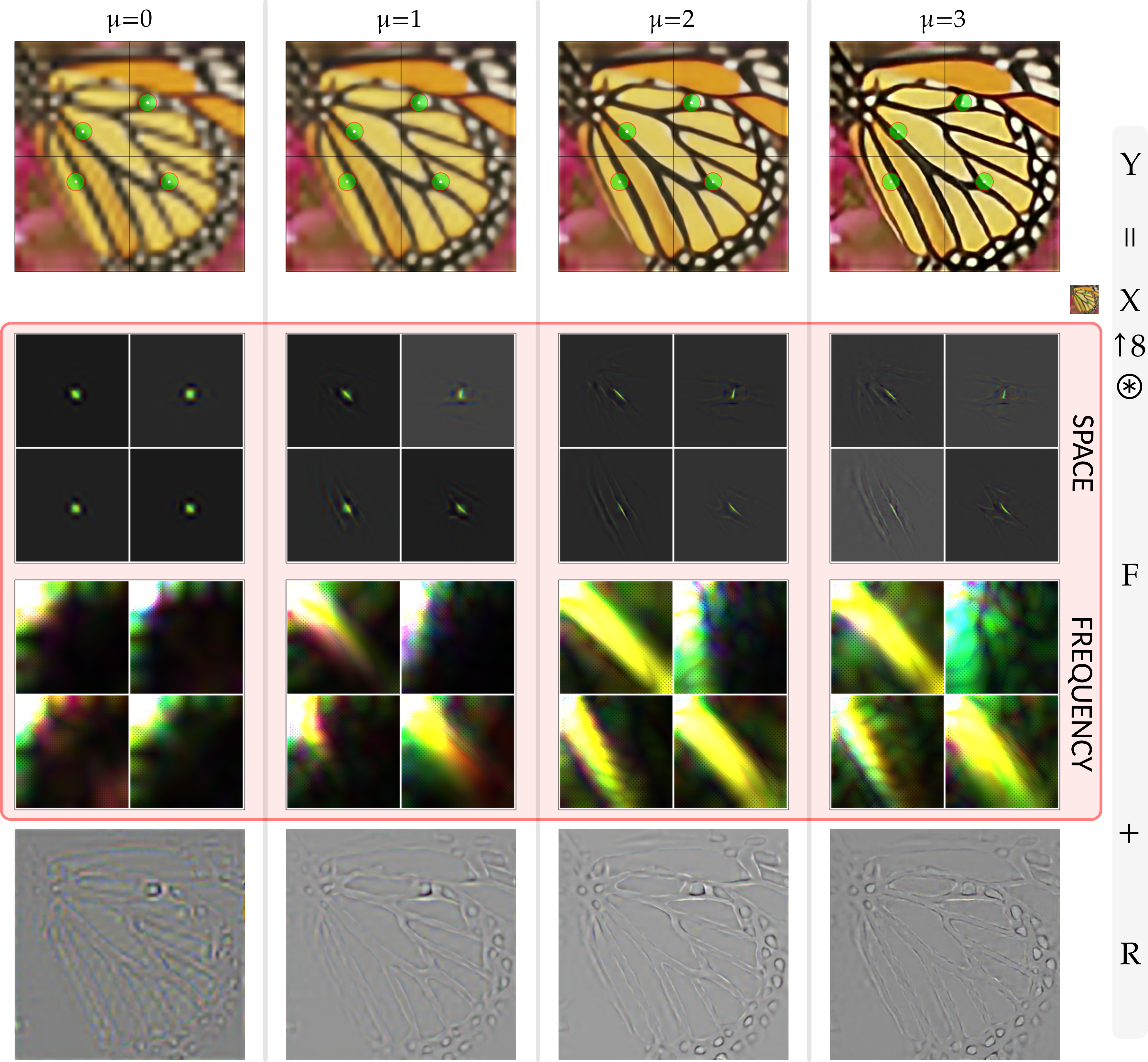}
    \caption{Output images for $8\times$ upscaling using different numbers of back--projections $\mu$. The network was trained with $\mu=2$ fixed. The effective filters are shown together with their frequency response (FFT). \label{fig:vis_sharp}}
\end{figure}
Figures \ref{fig:perceptual_4x} and \ref{fig:perceptual_8x} show the difference in perceptual quality for $4\times$ and $8\times$ upscaling. In general, EDSR and DBPN outputs look sharper and show consistent geometry. But often we find patches with aliasing, such as parallel lines, where the geometry becomes chaotic and perceptual quality behaves randomly. In such cases our system can take advantage compared to most of other systems.
\begin{figure}[ht]
  \centering
  \subfloat[$4\times$ results on Set14]{\includegraphics[width=0.9\linewidth]{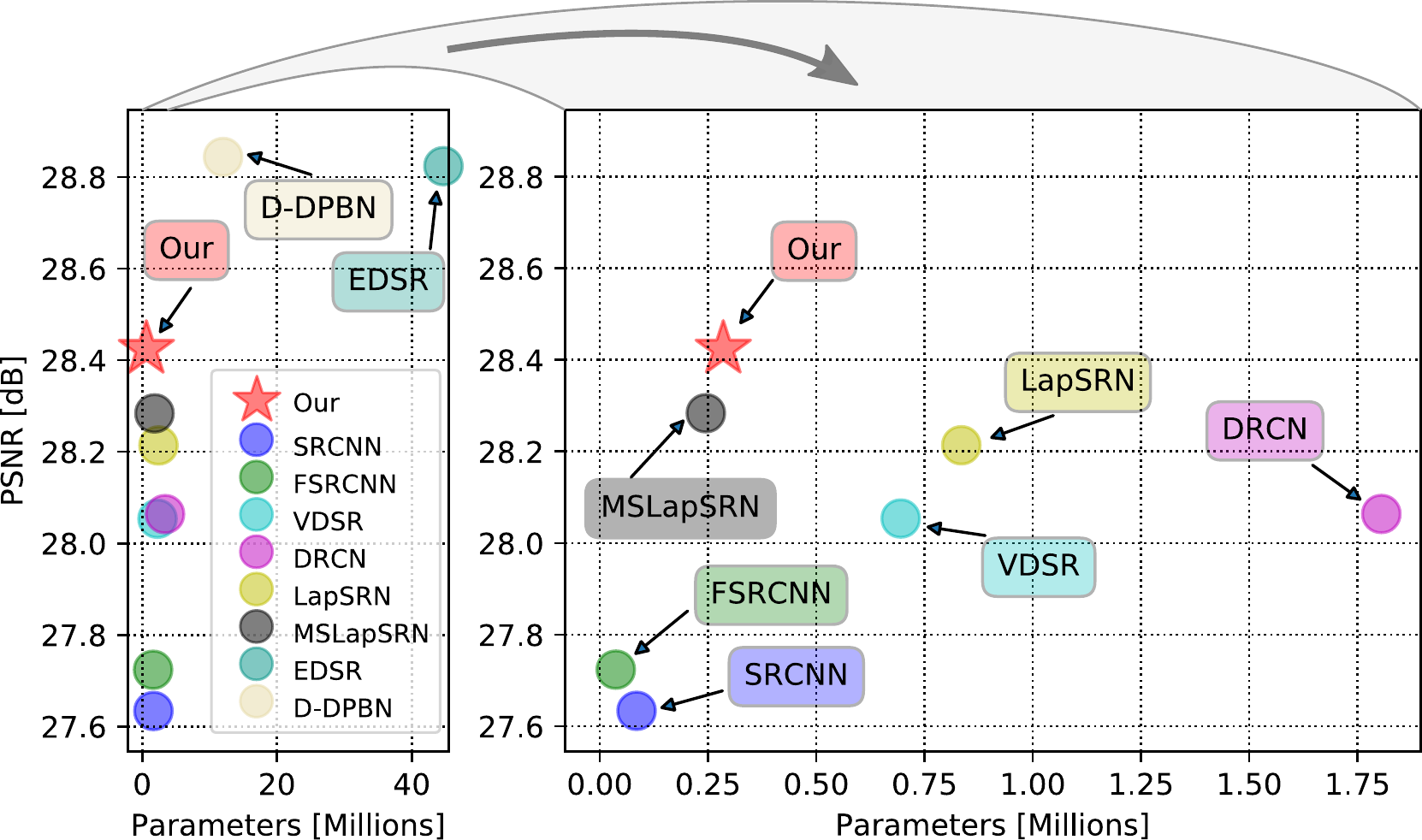}} \label{fig:scatter_4x} \\
  \subfloat[$8\times$ results on Set14]{\includegraphics[width=0.9\linewidth]{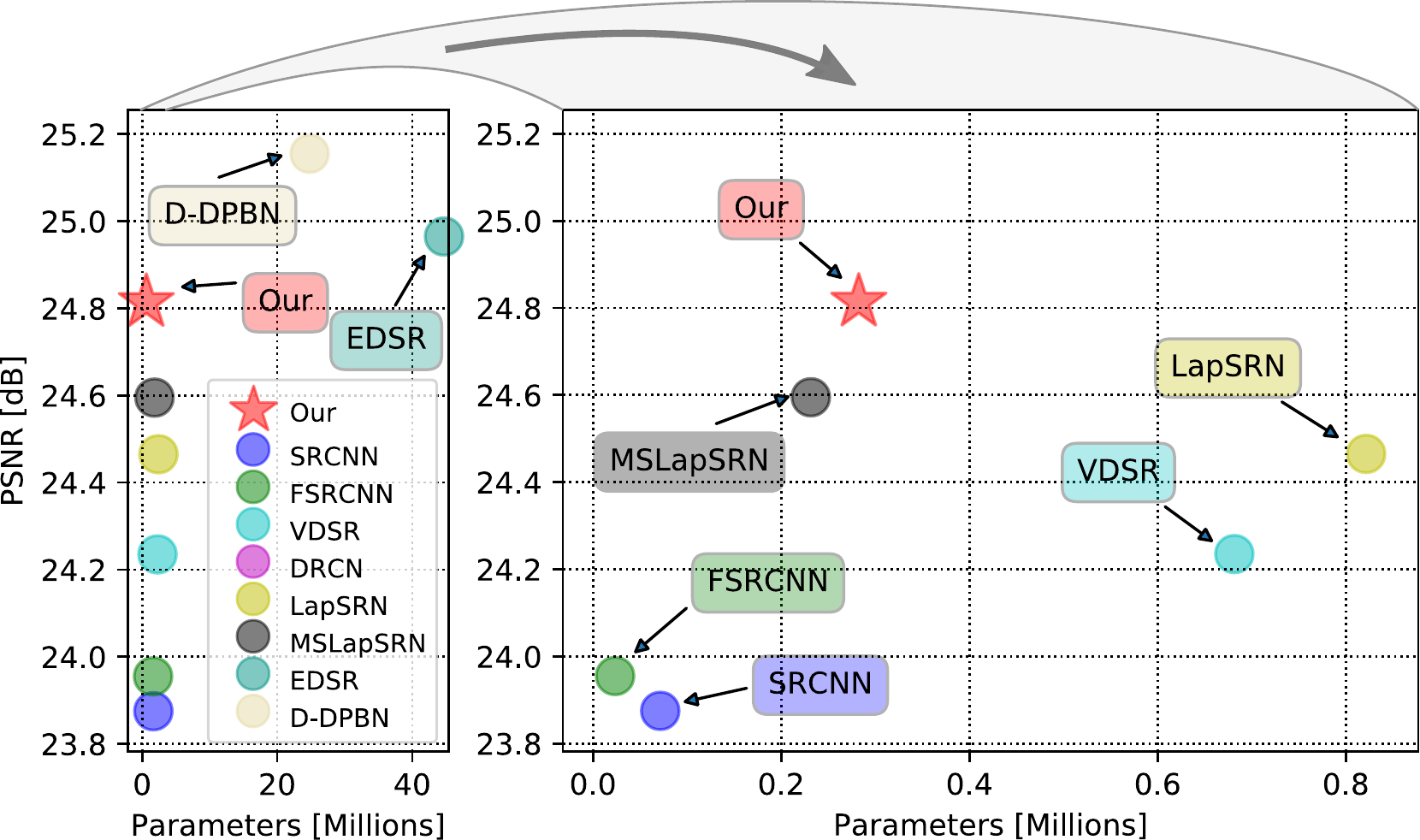}} \label{fig:scatter_8x}
    \caption{Quality vs Complexity of different SR methods. \label{fig:scatter}}
\end{figure}
\begin{figure*}[ht]
  \centering
  \includegraphics[width=\linewidth]{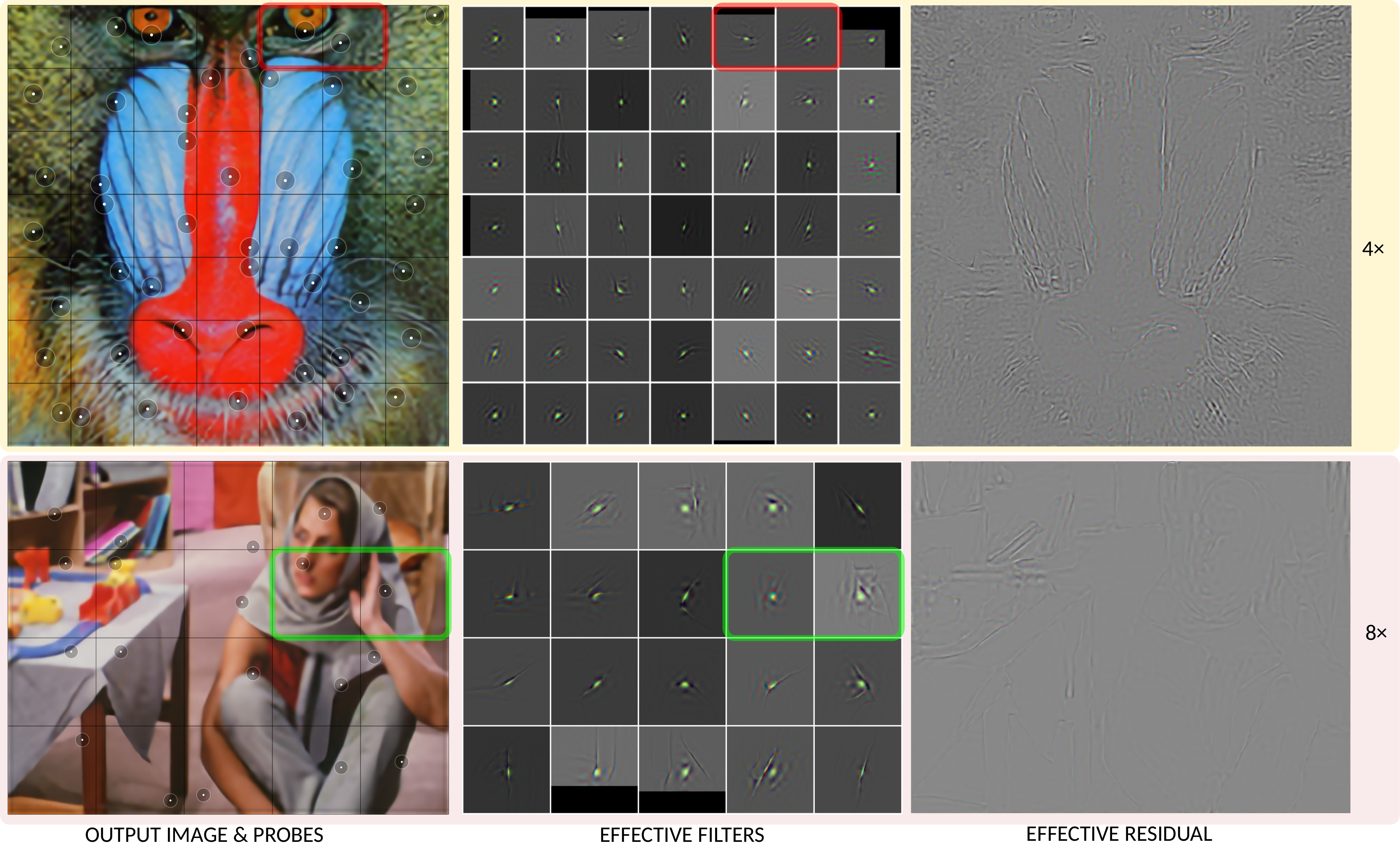}
    \caption{Effective filters and residual for $4\times$ and $8\times$ upscaler. \label{fig:vis_4x_8x}. Filters are not only directional but follow several features around a pixels. Residuals are in general small and thus most of the work is done by the filters.}
\end{figure*}

\subsection{Analysis}
\label{ssec:analysis}
Figure \ref{fig:vis_4x_8x} shows two examples of effective filters and residuals according to our novel visualization method. The shapes of filters following edges are reminiscent of the research on directional filter upscalers \cite{VRAlgazi_1991a,XLi_2001a}. In general, we observe that \textbf{filters are highly adaptive}. They become symmetric and narrow in flat areas (like Linear or Bicubic upscalers), directional close to sharp lines, and increase in size close to complex features. For example, in the $4\times$ upscale of Figure \ref{fig:vis_4x_8x} a red box shows an area with the eye of a baboon and the wrinkles around its eye. In the filter we can observe the shape of the eye and waves following the wrinkles. As expected, for $8\times$ upscaler the receptive field increases since the filters look bigger. We remind that we are using the same number of parameters, but passing through more layers. In the green box in \ref{fig:vis_4x_8x} we highlight the face and hand of a woman. The filter at the tip of the nose follows the face features, and the filter at the hand captures features from all the edges around.

Figure \ref{fig:vis_sharp} shows the effect of the depth on the effective filters and residuals. Here, training has tuned the network to work best with $\mu=2$ back--projections. For $\mu>2$ the output image looks oversharp, and too soft for $\mu<2$. Accordingly, filters are larger and extract more high--frequencies for $\mu>2$ and become smaller and low-pass for $\mu<2$. The same effect is observed with classic IBP \cite{Irani_1991a}, showing the effective design of our network architecture.

We observe that residuals are in general small and help fixing textures and small details. This is an indication that filters are doing most of the work for upscaling. We remind that, after freezing activations, the residual is a fixed component of the output that does not change with the input. Thus, residuals are very limited to estimate local details in the output as they depend only on activations for this purpose.

On the other hand, filters contain relations between neighboring pixels and we argue that because of this they are better to generalize. Effective filters also show all the details regarding the receptive field of the network. The receptive field is adaptive as the network does not use neighboring pixels if it does not need to (e.g. flat areas) and extend to large areas when the network is deep (e.g. $8\times$ upscaler) and local details are complex.

Finally, we remind that this analysis is precise. The network can be replaced by these adaptive filters plus residuals and it would give the exact same output. The strong dependency of filters and residuals on the input image shows all the non--linearities of the network. It is remarkable that convolutional networks can achieve the level of adaptivity revealed by these visualization experiments and it further justifies their success in super--resolution tasks.

\section{Conclusions}
\label{sec:conclusions}
We introduced a new architecture for single image super--resolution that reaches state of the art for methods with less that $2$ million parameters and a new technique to analyze the network. The analysis shows how the network learns to upscale by capturing complex relationships between pixels.

\bibliographystyle{aaai}
\bibliography{bibliography}

\onecolumn
\section{Supplementary Material}
\label{sec:appendix}
\subsection{Back--Projection Convergence}
\label{ssec:bp_convergence}

In the following analysis we consider vectors instead of images. This simplifies the analysis since we can use (rectangular) matrices $U$ and $D$ to represent \textbf{upsampling} and \textbf{downsampling} operators, respectively. Thus:
\begin{itemize}
    \item $Dy$ represents \textbf{downsampling} operation $(Y*g)\downarrow s$. Where $y=\text{vec}(Y)$ for a high resolution image $Y$.
    \item $Ux$ represents \textbf{upsampling} operation $(X\uparrow s)*p$. Where $x=\text{vec}(X)$ for a low resolution image $X$.
\end{itemize}

Our $L$--level downscaling model is now:
\begin{equation}
    x = D^{L}y \;.
\end{equation}
The target of back--projections is to recover this property in upscale images. Convergence is characterized by the mismatch error defined as follows:
\begin{defn}[Downscale mismatch error]
    The \emph{downscale mismatch error} of a high--resolution image $y_{L+1}$ at level $L+1$ is given by:
    \begin{equation}
        e(y_{L+1}) = x - D^L y_{L+1} \;.
    \end{equation}
\end{defn}
And the classic back--projection algorithm uses the following iteration:
\begin{defn}[Classic Back--projection]
    The one--level ($L=1$) back--projection iteration is given by:
    \begin{equation}
        y^{t+1}_2 = y^t_2 + Ue(y^t_2) \;.
    \end{equation}
\end{defn}

First, we recall the classic convergence theorem \cite{Irani_1991a}, in a simplified version adapted from \cite{Dai_2007a}.

\begin{thm}[Classic Back--Projection Convergence] \label{thm:classic_bp}
The one--level ($L=1$) back--projection iteration converges, $Dy\rightarrow x$, at exponential rate if the following condition holds:
\begin{equation}
    ||I-DU||_1 < 1 \;,
\end{equation}
where $||A||_1=\max_j \sum_i |a_{ij}|$.
\end{thm}
\begin{proof}
    \begin{eqnarray}
        e(y^{t+1}_2) & = & x - D y^{t+1}_2 \\
        & = & x - D (y^t_2 + Ue(y^t_2)) \\
        & = & \underbrace{x - D y^t_2}_{e(y^t_2)} - DUe(y^t_2) = (I-DU) e(y^t_2) \;,
    \end{eqnarray}
    By H\"{o}lder's inequality we get $||e(y^{t+1}_2)||_1 \leqslant ||I-DU||_1 \; ||e(y^t_2)||_1$ which proves the theorem.
\end{proof}

Now, we rewrite the MGBP algorithm using matrix/vector notation:

\begin{algorithm*}[th]
    \centering
    \begin{tabular}{ll}
        $MGBP(x)$: & $BP^{\mu}_{k}(u, y_1,\ldots,y_{k-1})$: \\

        \resizebox{.5\textwidth}{!}{
            \begin{minipage}{.6\textwidth}
                \begin{algorithmic}[1]
                    \REQUIRE Input image $x$.
                    \REQUIRE Integers $\mu\geqslant 0$ and $L\geqslant 1$.
                    \ENSURE Images $y_k$, $k = 2,\ldots,L$.

                    \STATE $y_1 = x$
                    \FOR{$k = 2,\ldots,L$}
                        \STATE $u = U y_{k-1}$
                        \STATE $y_{k} = BP^{\mu}_{k}\left(u, y_1,\ldots,y_{k-1} \right)$
                    \ENDFOR
                \end{algorithmic}
        \end{minipage}
        }
        &
        \resizebox{.5\textwidth}{!}{
            \begin{minipage}{0.6\textwidth}
                \begin{algorithmic}[1]
                    \REQUIRE Image $u$, level index $k$, number of steps $\mu$.
                    \REQUIRE Images $y_1,\ldots,y_{k-1}$ (only for $k>1$).
                    \ENSURE Updated image $u$

                    \IF{$k > 1$}
                        \FOR{$step = 1,\ldots,\mu$}
                            \STATE $d = BP^{\mu}_{k-1}\left( D u, y_1,\ldots,y_{k-2} \right)$
                            \STATE $u = u + U(y_{k-1} - d)$
                        \ENDFOR
                    \ENDIF
                \end{algorithmic}
            \end{minipage}
        }
    \end{tabular}
    \caption{Multi--Grid Back--Projection (MGBP) algorithm in matrix/vector notation.} \label{tab:alg}
    \label{alg:mgbp_classic}
\end{algorithm*}

\begin{defn}[single--level mismatch error]
    Assume that the MGBP algorithm runs $\mu$ back--projection steps from level $k=2,\ldots,L$. Then, the \emph{single--level mismatch error} for a $L+1$--resolution image $y_{L+1}$ is given by:
    \begin{equation}
        q(y_{L+1}) = y^\mu_L - D y_{L+1}
    \end{equation}
\end{defn}

This is usefull since the MGBP iteration can be written as $y^{t+1} = y^t + Uq(y^t)$. It easily follows that:
\begin{eqnarray}
    q(y^{t+1}_{L+1}) & = & y^\mu_L - D y^{t+1}_{L+1} \\
    & = & y^\mu_L - D (y^t_{L+1}+Uq(y^t_{L+1})) \\
    & = & (I-DU) q(y^t_{L+1})
\end{eqnarray}

\begin{thm}[MGBP Convergence] \label{thm:mgbp}
The MGBP iteration converges, $D^L y_{L+1}\rightarrow x$, at exponential rate if the following condition holds:
\begin{equation}
    ||I-DU||_1 < 1 \;.
\end{equation}
\end{thm}
\begin{proof}
    We follow by induction. From the classic convergence theorem we know that $||e(y^{t+1}_2)||_1 \leqslant ||I-DU||_1 \; ||e(y^t_2)||_1$. For $L=1$ the MGBP iteration is equivalent to classic back--projection, and so the theorem is verified for $L=2$.

    Next, we assume that the theorem holds for $L$ and we study the error for $L+1$:
    \begin{eqnarray}
        e(y^\mu_{L+1}) & = & x - D^L y^\mu_{L+1} \\
        & = & x - D^{L-1} (y^\mu_L - q(y^\mu_{L+1})) \\
        & = & x - D^{L-1} (y^\mu_L - (I-DU)^\mu q(y^0_{L+1})) \\
        & = & x - D^{L-1} (y^\mu_L - (I-DU)^\mu (y^\mu_L - DUy^\mu_L)) \\
        & = & \underbrace{x - D^{L-1} y^\mu_L}_{e(y^\mu_L)} + D^{L-1} (I-DU)^{\mu+1} y^\mu_L
    \end{eqnarray}

    By triangle inequality and H\"{o}lder's inequality we have:
    \begin{equation}
        ||e(y^\mu_{L+1})||_1 \leqslant ||e(y^\mu_L)||_1 + ||D^{L-1}||_1 ||(I-DU)^{\mu+1}||_1 ||y^\mu_L||_1 \;.
    \end{equation}
    The first term decreases exponentially by inductive assumption. The second term also decreases exponentially with $\mu$ if $||I-DU||_1 < 1$. The term $||y^\mu_L||_1$ is bounded since by inductive assumption we know that $y^\mu_L$ converges to a fixed $L$--level image.
\end{proof}

\end{document}